\documentclass[runningheads,envcountsame]{llncs}
\usepackage[T1]{fontenc}
\usepackage[dvipsnames]{xcolor}
\usepackage{amsmath}    
\usepackage{amssymb}
\usepackage{graphicx}
\usepackage{subcaption}
\usepackage[ruled,vlined,linesnumbered]{algorithm2e}
\SetAlgoSkip{0pt}        
\SetAlgoInsideSkip{0pt}  
\SetInd{0.3em}{0.5em}   
\SetNlSty{}{}{}         
\SetAlgoNlRelativeSize{-1} 
\usepackage{enumerate}
\usepackage{cite}
\usepackage{tikz}
\usepackage{pifont}%

\usepackage{thmtools}
\usepackage{thm-restate} 

\usetikzlibrary{arrows}
\usetikzlibrary{shapes}
\usetikzlibrary{automata}
\usetikzlibrary{positioning,shadows,trees}
\tikzset{initial text={}}

\usepackage{booktabs} 
\usepackage{siunitx} 
\usepackage{tabularx}

 \usepackage{cleveref} 
\Crefname{figure}{Fig.}{Figs.}
\usepackage{lineno}

\usepackage{ltablex} 
\keepXColumns 

\usepackage[inline]{enumitem}
\setlist[enumerate]{topsep=2pt,itemsep=0pt}
\setlist[itemize]{topsep=2pt,itemsep=0pt}

\newcommand{\ra}{RA\xspace}
\newcommand{\dra}{DRA\xspace}
\newcommand{\wdra}{DRA\xspace}

\newcommand{\fullhdm}{hyper-data-minimal\xspace}

\newcommand{\fullhsm}{hyper-state-minimal\xspace}

\newcommand{\fullhm}{hyper-minimal\xspace}

\newcommand{\A}{\mathcal{A}}

\newcommand{\lenTA}{\textit{len}^2_\mathcal{A}}

\newcommand{\alphabet}{\Sigma}
\newcommand{\finwords}{\alphabet^*}

\newcommand{\lang}{L}

\newcommand{\emptyword}{\varepsilon}

\newcommand{\bigO}{\mathcal{O}}

\newcommand{\alphabetQ}{\mathbb{Q}}

\newcommand{\naturals}{\mathbb{N}}

\newcommand{\autA}{\A}
\newcommand{\autB}{\mathcal{B}}

\newcommand{\symdiff}{\ominus}
\newcommand{\merge}{\texttt{merge}}

\newcommand{\datamin}{\texttt{hyper-data-minimize}}

\newcommand{\mdot}{\!\cdot\!}
\iftrue
\usepackage{todonotes}
	\newcommand{\warn}[1]{\todo[inline,color=red!75,caption={W}]{\textbf{IMPORTANT: #1}}}
	\newcommand{\dd}[1]{\todo[inline,color=teal!10,caption={DiDe}]{\textbf{Di-De:} #1}}
	\newcommand{\qiyi}[1]{\todo[inline,color=Orchid!10,caption={Qiyi}]{\textbf{Qiyi:} #1}}
    \newcommand{\ly}[1]{\todo[inline,color=orange!10,caption={LY}]{\textbf{LY:} #1}}
     
\else
\newcommand{\warn}[1]{}
\newcommand{\dd}[1]{}
\newcommand{\qiyi}[1]{}
\newcommand{\ly}[1]{}

\fi

\begin{document}\sloppy
\title{Hyper-Minimization for Deterministic Register Automata}
%
%
\author{Yong Li\inst{1}\orcidID{0000-0002-7301-9234} \and
Qiyi Tang\inst{2}\orcidID{0000-0002-9265-3011} \and
Di-De Yen\inst{2}\orcidID{0000-0003-0045-9594}}

\authorrunning{Y. Li, Q. Tang and D. Yen} 

%
%
\institute{
Key Laboratory of System Software (Chinese Academy of Sciences), 
Institute of Software Chinese Academy of Sciences, Beijing, China\\
\email{liyong@ios.ac.cn} \and
School of Computer Science and Informatics, University of Liverpool, UK\\
\email{\{qiyi.tang,d.d.yen\}@liverpool.ac.uk}}
\maketitle              
%


\begin{abstract}


We investigate hyper-minimization for deterministic register automata (DRAs). 
We begin by introducing DRA counterparts of classical notions from deterministic finite automata. 
Building on these foundations, we present an algorithm for hyper-minimizing well-typed DRAs, where each state is associated with a unique register type. 
The resulting automata are minimal with respect to both the number of states and registers among all well-typed DRAs. 
We prove the correctness of the proposed algorithm, thereby establishing the decidability of hyper-minimization for well-typed DRAs.

\end{abstract}
%
%
%




\section{Introduction}


The minimization problem is a fundamental topic in automata theory, dating back to~\cite{Nerode:57}, where an automaton is minimal if it has no smaller equivalent counterpart. For finite automata (FAs), minimality corresponds to having the smallest number of states among all equivalent FAs. While every FA can be minimized, the general problem is PSPACE-complete~\cite{Jiang:93:SIAM}; in contrast, deterministic finite automata (DFAs) admit an $\mathcal{O}(n\log n)$ minimization algorithm~\cite{Hopcroft:07:text}, where $n$ is the number of states.

For register automata (RAs), which extend FAs with registers for storing and comparing data values, minimality requires minimizing both the number of states and registers. The minimization problem for deterministic register automata (DRAs) over dense ordered and unordered domains is decidable~\cite{Ley:10:TR}.

The notion of DFA minimization was generalized to \emph{hyper-minimization} in~\cite{DBLP:journals/ita/BadrGS09}. A DFA $\autA'$ is hyper-minimal for a DFA $\autA$ if the symmetric difference between their languages is finite. For example, let $\autA$ recognize $L_{\textit{even}+x^{49}}$ over $\Sigma=\{a,b\}$, consisting of all even-length words together with $a^{49}$ and $b^{49}$. Any DFA for this language requires at least 101 states: one initial state, 49 states each to count $a^{49}$ and $b^{49}$, and two states to track parity. In contrast, a hyper-minimal DFA needs only two states, recognizing $L_{\textit{even}}$, which consists of all even-length words. This illustrates that allowing finitely many errors can yield significantly more succinct automata. An $\mathcal{O}(n\log n)$ algorithm for DFA hyper-minimization was later given in~\cite{Holzer:10:CIAA}, matching the complexity of standard minimization.

Hyper-minimization has been studied for several extensions of DFAs, including deterministic tree automata~\cite{Maletti:12:CIAA} and deterministic weighted automata~\cite{Maletti:11}. In contrast, it has not been systematically investigated for automata over infinite alphabets, such as DRAs or symbolic finite automata. One reason is that hyper-minimization is often considered less meaningful in this setting, since over infinite alphabets any non-empty symmetric difference is typically infinite \cite{Veanes:14:POPL}.

For example, consider extending $L_{\textit{even}+x^{49}}$ to the alphabet $\mathbb{Q}$, yielding the language of all even-length words together with all words of the form $c^{49}$ for $c \in \mathbb{Q}$. For any RA $\autA$, if $c^{49} \in \lang(\autA)$ for some $c \in \mathbb{Q}$, then $d^{49} \in \lang(\autA)$ for all $d \in \mathbb{Q}$. Hence, if $\autA$ does not recognize $L_{\textit{even}+x^{49}}$ exactly, the symmetric difference $L_{\textit{even}+x^{49}} \ominus \lang(\autA)$ is necessarily infinite.


Nevertheless, if $\lang(\autA)$ consists of all even-length words over $\mathbb{Q}$, then all words in $L_{\textit{even}+x^{49}} \ominus \lang(\autA)$ share the same \emph{word type}, namely words of the form $c^{49}$ for $c \in \mathbb{Q}$. Intuitively, two words have the same word type if they induce the same order relations among their positions. For example, the words $u = 3 \cdot 1 \cdot 5$ and $v = \frac{5}{2} \cdot 2 \cdot 9$ have the same word type over $\mathbb{Q}$ under the standard order $<$, as the relative ordering between any pair of positions is identical in both words.

Motivated by this observation, we study hyper-minimization of DRAs modulo word types. A DRA $\autA'$ is hyper-minimal for a given DRA $\autA$ if $\lang(\autA) \ominus \lang(\autA')$ contains only \emph{finitely many word types}. Under this notion, the hyper-minimal DRA for $L_{\textit{even}+x^{49}}$ over $\mathbb{Q}$ recognizes $L_{\textit{even}}$ and requires only two states.

In general, a minimal DRA need not exist due to a trade-off between the number of states and registers (see \Cref{app:ex_trade-off}); the same issue arises for hyper-minimization. However, for \emph{well-typed} DRAs—where each state has a unique register type—a minimal DRA is guaranteed to exist~\cite{Ley:10:TR}. We show that an analogous result holds for hyper-minimization, and present an algorithm for hyper-minimizing well-typed DRAs, thereby establishing its decidability.

\section{Preliminaries}

We use $\mathbb{R}$ (resp.\ $\mathbb{Q}$, $\mathbb{Z}$, $\mathbb{N}$) to denote the sets of real (resp.\ rational, integer, and non-negative) numbers.
A set $S \subseteq \mathbb{R}$ is \emph{dense} in $\mathbb{R}$ if, for every $s < t$, there exists $r \in S$ with $s < r < t$~\cite{Royden:10:book}. Accordingly, $\mathbb{Q}$ is dense in $\mathbb{R}$, whereas $\mathbb{Z}$ is not. 


Let $\Sigma$ be an alphabet and $R$ a binary relation on $\Sigma$, where $\Sigma$ may be finite or infinite. A \emph{(data) word} is a finite sequence over $\alphabet$.
For words $u = a_1 \dots a_m$, $v = b_1 \dots b_n$ in $\Sigma^*$ and a symbol $d \in \Sigma$, we write $d \in u$ if $d = a_i$ for some $1 \le i \le m$, $u \cup v$ as the set $\{a_1, \cdots, a_m, b_1,\cdots, b_n\}$ and $u \cdot v$ as the concatenation $a_1 \dots a_m b_1 \dots b_n$. Given two symbols $a$ and $b$, we use $u[a/b]$ to denote the word obtained by replacing each $a$ in $u$ with $b$.
The length of $u$ is denoted by $|u|$.
The relation $R$ induces an equivalence relation $\sim_R$ on $\Sigma^*$: for $u = a_1 \dots a_m$ and $v = b_1 \dots b_n$, we say $u$ and $v$ are in the same equivalence $\tau$ iff (1) $m = n$, and (2) $(a_i, a_j) \in R$ if and only if $(b_i, b_j) \in R$ for all $1 \le i, j \le n$.
Each equivalence class $\tau$ of $\sim_R$ is called a \emph{$\sim_R$-word type}. 
$\tau$ has length $n$ if $|\tau| = n$. 
A language $L \subseteq \Sigma^*$ is a \emph{data language} over $(\Sigma, R)$ if, for all words $u, v$ of the same word type, $u \in L \Leftrightarrow v \in L$. 
In this paper, we consider only two kinds of domains: dense ordered domains or $R$ is the identity relation.

Given two data languages $L$ and $L'$, we say that $L$ and $L'$ are \emph{almost-equal}
if there exists a \emph{finite} set $T$ of word types such that every word in the symmetric difference $L \ominus L'$ has a word type $\tau \in T$. In other words, two almost-equal data languages differ only on finitely many word types.

\emph{Register automata} (\ra{s}) extend finite-state automata to infinite alphabets and recognize data languages. There are several equivalent definitions in the literature~\cite{Kaminski:94:TCS, Neven:04:ACMCL, Ley:10:TR, DBLP:conf/mfcs/BalachanderFGT25}. We adopt the definition of \ra{s} from~\cite{Ley:10:TR} in this paper.

\begin{definition}\label{def:RA}
Given $k \in \mathbb{N}$, an alphabet $\Sigma$, and a binary relation $R$ on $\Sigma$, 
a $k$-register automaton ($k$-\ra) $\A$ over $(\Sigma,R)$ 
is a tuple $(Q, q_0, F, \Delta)$~where:
$Q$ is a set of states, partitioned into disjoint subsets $Q_0, \dots, Q_k$,
$q_0 \in Q_0$ is the \emph{initial state},
$F \subseteq Q$ is the set of \emph{final states}, and
$\Delta$ is a finite set of \emph{transitions} of the form $(p, \tau, E, q)$, 
where $p \in Q_i$ and $q \in Q_j$ for some $0 \le i, j \le k$, 
$\tau$ is a $\sim_R$-word type of length $i+1$, 
and $E \subseteq \{1, \ldots, i+1\}$.
\end{definition}

An \ra $\A$ is \emph{deterministic} (\dra) if, for any two transitions $(p, \tau, E, q)$ and $(p', \tau', E', q')$, we have $(E, q) = (E', q')$ whenever $(p, \tau) = (p', \tau')$.

A \emph{configuration} of $\A$ is a pair $(q, v)$, where $q \in Q_i$ and $v$ is a word over $\Sigma$ of length $i$, for some $0 \le i \le k$. 
We say a configuration $(q, v)$ is \emph{accepting} if $q \in F$.
For configurations $(p, u)$ and $(q, v)$ and a symbol $a \in \Sigma$, we write $(p, u) \xrightarrow{a}_\A (q, v)$ or $(p, u) \xrightarrow{\tau: E}_\A (q, v)$ if there exists a transition $\delta = (p, \tau, E, q)$ such that $u \cdot a = a_1 \dots a_n$ is of type $\tau$ and $v$ is obtained from $u \cdot a$ by removing all $a_i$ with $i \in E$, where (1) $E \neq \emptyset$ if $|u| = k$, preventing memory overflow, and
(2) $j \in E$ whenever $1 \le j < n$ and $a_j = a_n (=a)$, ensuring no data duplication.
The subscript $\A$ is omitted when clear from context.

Let $w = b_1 \dots b_m \in \Sigma^*$. A sequence of configurations $\pi = (p_0, u_0) \dots (p_m, u_m)$ is a \emph{run} on $w$, written $(p_0, u_0) \xrightarrow{w} (p_m, u_m)$, if $(p_i, u_i) \xrightarrow{b_i} (p_{i+1}, u_{i+1})$ for all $0 \le i < m$. It is \emph{accepting} if $(p_0, u_0) = (q_0, \emptyword)$ and $(p_m, u_m) $ is accepting.
We say $\autA$ is \emph{complete} if every word $u \in \finwords$ has a run in $\autA$ from $(q_0, \epsilon)$. Any \dra can be made complete by introducing a \emph{sink state}.
The language recognized by $\autA$ is $\lang(\autA) = \{ w \in \Sigma^* \mid (q_0, \epsilon) \xrightarrow{w} (q_f, u), q_f \in F, u \in \Sigma^* \}$, which is a data language over $(\Sigma, R)$. Two \dra{s} $\autA$ and $\autA'$ are \emph{equivalent} if $\lang(\autA) = \lang(\autA')$.
Given a configuration $c$ of $\autA$, we use $\autA_c$ to denote the \ra obtained from $\autA$ such that all accepting runs start from configuration $c$. Note that if $\autA$ is a finite automaton, then a configuration is simply a state. 

A state $q$ of $\autA$ is \emph{useful} if there exist words $w, w'$ such that $(q_0, \epsilon) \xrightarrow{w}_\autA (q, u)$ and $(q, v) \xrightarrow{w'}_{\autA} (q_f, x)$ for some $u, v, x$ and an accepting state $q_f$. $\autA$ is \emph{trimmed} if all its states, excluding the sink state, are useful. In this paper, we primarily consider \dra{s} over dense or non-ordered domains (i.e., with $R$ being identity test). Since the reachability problem for these models is decidable \cite{Yen:14:Infinite, Ley:10:TR}, we assume all \dra{s} in the remainder of this paper are trimmed.


\begin{figure}[t]
    \centering
    \includegraphics[width=0.8\linewidth]{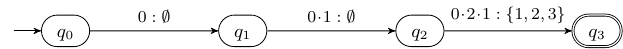}
    \caption{An example \dra $\autA$}
    \label{fig:DRA_mid}
\end{figure}

\begin{example}\label{ex:mid}
Let $L_{\textit{mid}} = \{a_1a_2a_3 \mid a_1 < a_3 < a_2\}$ over $(\mathbb{Q},<)$. The \dra in Fig.~\ref{fig:DRA_mid}, with initial state $q_0$ and accepting state $q_3$, recognizes $L_{\textit{mid}}$. Transitions are labeled by $\tau : E$. For instance, the transition from $q_2$ to $q_3$ is labeled $0\mdot 2\mdot 1 : \{1,2,3\}$.
On input $w = 3\cdot 9\cdot 7$, the automaton has the accepting run $(q_0,\emptyword) \xrightarrow{3} (q_1,3) \xrightarrow{9} (q_2,3\cdot 9) \xrightarrow{7} (q_3,\emptyword)$, showing that $w \in L_{\textit{mid}}$.
\qed
\end{example}

\begin{definition}
Let $\A$ be an \ra over $(\Sigma, R)$. We call $\A$ \emph{well-typed}  if for every two configuration transitions $(p, u) \xrightarrow{a}_{\A} (q, v)$ and  $(p', u') \xrightarrow{a'}_{\A} (q, v')$ ending in the same state $q$, we have $v$ and $ v'$ are of the same word type $\tau$. 
\end{definition}

Intuitively, an \ra is well-typed if the word type $\tau$ of its registers is uniquely determined by the current state $p$ and we say $p$ is of \emph{$\tau$-register-type}.

Let $L$ be a data language over the data domain $D$, and let $\mathcal{A}$ and $\mathcal{B}$ be two \ra{s} over $D$. We say that $\mathcal{A}$ and $\mathcal{B}$ are \emph{almost-equivalent} (resp., \emph{equivalent})
if the data languages $L(\mathcal{A})$ and $L(\mathcal{B})$ are almost-equal (resp., equal); furthermore, $\mathcal{A}$ \emph{almost-recognizes} $L$ if $L(\mathcal{A})$ is almost-equal to $L$.
Suppose that $\mathcal{A}$ has $n$ states and $k$ registers. We say that $\mathcal{A}$ is \emph{\fullhdm} (resp., \emph{data-minimal}) for $L$ if every \ra that almost-recognizes (resp., recognizes) $L$ has at least $k$ registers, and it is \emph{\fullhsm} (resp., \emph{state-minimal}) for $L$ if every \ra that almost-recognizes (resp., recognizes) $L$ has at least $n$ states. Finally, $\mathcal{A}$ is \emph{\fullhm} (resp., \emph{minimal}) if it is both \fullhdm and \fullhsm (resp., both data-minimal and state-minimal).


Given a directed graph $G$ with an initial node $v_0$, a node $v$ of $G$ is a \emph{preamble} if there are finitely many paths from $v_0$ to $v$, otherwise it is a \emph{kernel}. 
Intuitively, $v$ is a preamble if no path from $v_0$ to $v$ contains a loop.
A state of an automaton $\mathcal{A}$ is a \emph{preamble state} if it is a preamble node in the underlying graph with the initial state as the initial node. The \emph{kernel states} of $\mathcal{A}$ are defined analogously.


\section{Characteristics of Well-Typed \dra{s} 
}
\label{sec:dra_properties}

An important property of \dra{s} over dense or non-ordered domains is the decidability of the minimization problem \cite{Ley:10:TR}. The derivation of this result in \cite{Ley:10:TR} is analogous to that of DFA minimization, which is closely related to the Nerode congruence.

For regular languages, two words are Nerode-equivalent if no distinguishing extension exists. For data languages, however, equivalence additionally depends on \emph{memorability}~\cite{Ley:10:TR}, which captures the symbols that influence future acceptance. Intuitively, a symbol is $L$-memorable in a word $u$ if modifying it (while preserving type) can change the membership of some extension of $u$.
Consider the language $L_{\textit{mid}}$ in \Cref{ex:mid}. For any \dra{} $\mathcal{A}$ recognizing $L_{\textit{mid}}$, after reading $u = 2 \cdot 6$, the values $2$ and $6$ must be stored in registers to check whether the next input lies between them. Hence, both $2$ and $6$ are \emph{$L$-memorable} in $u$. 
The formal definition follows:

\begin{definition}\label{def:memorable}
Let $L$ be a language over $(\Sigma, R)$. A symbol $a \in \Sigma$ is \emph{$L$-memorable} in a word $u \in \Sigma^*$ if there is a symbol $b \in \Sigma$ and a word $w \in \Sigma^*$ such that
(1) $u\cdot w \sim_R (u\cdot w)[a/b]$; and
(2) $uw\in L \Leftrightarrow u\cdot w[a/b]\not\in L$.
\end{definition}
Intuitively, a letter $a$ in $u$ is memorable if it can distinguish between continuations $w$ and $w[a/b]$ such that $uw \sim_R (uw)[a/b]$, yet replacing $a$ with a nearby symbol $b$ in $w$ changes the acceptance of $uw$ versus $uw[a/b]$.

We write $mem_L(u)$ for the sequence of $L$-memorable symbols in $u$. 
With the notion of memorability, we can then define the Nerode congruence relation $\cong_L$ and the canonical automaton $\mathcal{C}_L$ for a data language $L$. For details, see \cite{Ley:10:TR} or \Cref{app:myhill_nerode}. As a result, we have the Myhill-Nerode theorem for \dra{s}:

\begin{restatable}{theorem}{thmMyhill}\cite{Ley:10:TR}\label{thm:myhill}
A data language $L$ is recognized by a \dra iff $\cong_L$ has finite index. 
Moreover, every \dra-recognizable language $L$ has a canonical automaton $\mathcal{C}_L$ which is the unique minimal well-typed \dra up to isomorphism\footnote{This claim does not hold when the well-typeness constraint is dropped.}.
\end{restatable}

The well-typedness condition in \Cref{thm:myhill} is necessary and sufficient for the existence and uniqueness of the minimal \dra{}. Without it, a trade-off between the number of states and registers may arise (see \Cref{app:ex_trade-off}). In the remainder of this paper, we assume all \dra{s} are well-typed unless stated otherwise.

Lastly, we present a pumping property for \dra{s} that is employed in our hyper-minimization algorithms.

For a finite automaton $\autA$ with $n$ states, any run $\pi$ of length greater than $n$ contains a non-empty subrun that can be pumped or depumped. This is the classical Pumping Lemma~\cite{Hopcroft:07:text}, which can be extended to multiple runs. In particular, given two runs $\pi_1$ and $\pi_2$ of $\autA$ over the same word $w$ with $|w| > n^2$, there exist non-empty subruns of $\pi_1$ and $\pi_2$ that can be pumped or depumped simultaneously. We establish a similar property for \dra{s}:

\begin{restatable}{lemma}{LemPumpTwoRuns}\label{lem:pump_two_runs}

    Let $\autA$ be a \dra{} over $(\Sigma,R)$ with $n$ states and $k$ registers, and let $\lenTA = n^2 \cdot (2k)!$. Given a word $w$ and two runs $\pi$ and $\rho$ from configurations $(p,u)$ and $(q,v)$ over $w$, ending in states $r$ and $s$, respectively, there exists a word $w'$ with $|w'| \leq \lenTA$ and corresponding runs $\pi'$ and $\rho'$ from $(p,u)$ and $(q,v)$ that also end in $r$ and $s$.
    Moreover, if $|w| > \lenTA$, then for every $\ell \in \mathbb{N}$, there exists a word $w''$ with $|w''| > \ell$ and corresponding runs $\pi''$ and $\rho''$ from $(p,u)$ and $(q,v)$ that end in $r$ and $s$.
\end{restatable}

\section{Hyper-Minimization}\label{sec:hyper-minimization}

Before introducing our hyper-minimization algorithm for \dra{s}, we briefly recall hyper-minimization for DFAs, which are a special class of \dra{s} over a finite alphabet and without registers.

We begin with the \emph{merge} operation from DFA hyper-minimization. For a DFA $\autA$, a state $p$ can be merged into a distinct state $q$ if:
\begin{itemize}
\item $p$ is a \emph{preamble} state, i.e., no loops in any path from the initial state to $p$;
\item $p$ and $q$ are \emph{almost-equivalent}, i.e., $\lang(\autA_p)\ominus\lang(\autA_q)$ is finite.
\end{itemize}

We denote by $\merge(p \rightarrow q, \autA)$ the DFA obtained by performing the merge operation, i.e., redirecting all incoming transitions of $p$ to $q$.
Let $\autB = \merge(p \rightarrow q, \autA)$.
It is not difficult to see that $\lang(\autA)$ is almost-equivalent to $\lang(\autB)$. In particular, $\lang(\autA) \symdiff \lang(\autB) = \{u\cdot w \mid q_0 \xrightarrow{u} p, w \in \lang(\autA_p)\symdiff \lang(\autA_q)\}$.
This symmetric difference set is finite since $p$ is a preamble state—hence only finitely many prefixes $u$ lead to $p$—and $\lang(\autA_p)\symdiff \lang(\autA_q)$ is finite by almost-equivalence definition.
By recursively applying the merge operation, we eventually obtain a minimal DFA that is almost-equivalent to $\autA$.

The basic DFA hyper-minimization algorithm (cf. \Cref{algo:dfa-hyper-minimization} in \Cref{app:dfa-hypermin}) proceeds as follows: it first minimizes the input DFA and then recursively merges a preamble state to its almost-equivalent state, where $p \approx q$ denotes that $\lang(\autA_p)$ and $\lang(\autA_q)$ are almost-equivalent.

The most interesting part, which is also missed so far, is the computation of the almost-equivalence relation on the states in a DFA $\autA$.
Existing algorithms compute the almost-equivalence relation based on the following result:
\begin{lemma}[\cite{DBLP:journals/ita/BadrGS09}]\label{lem:dfa-word-bound}
    Let $\autA$ be a minimal DFA with $n$ states. Then, two distinct states $p, q$ are almost-equivalent, i.e., $p \approx q$ if and only if there is an integer $\ell > 0$ such that $\autA_p(w) = \autA_q(w)$ for all $w \in \finwords$ with $|w| \geq \ell$, where $\A_r(w)$ is the state reached over $w$ from the initial state for $r=p,q$.
\end{lemma}

Let $n$ be the number of states in $\A$. 
We can restrict $\ell$ in \Cref{lem:dfa-word-bound} to be at most $n^2$. 
To see this, suppose that for every $k \ge n^2$ there exists a word $w$ with $|w| \geq k$ such that $\A_p(w) \neq \A_q(w)$; that is, $p' = \A_p(w) \neq  \A_q(w) = q'$.
Since $\A$ is minimal, this implies that $\lang(\A_{p'}) \neq \lang(\A_{q'})$.
Hence, there exists a word $w' \in \lang(\A_{p'}) \symdiff \lang(\A_{q'})$.
Because $|w| \geq n^2$, the word $w$ can be decomposed as $w = xuy$ such that $(\A_p(x), \A_q(x)) = (\A_p(xu), \A_q(xu))$.
It then follows that $xu^i y w' \in \lang(\A_p) \symdiff
\lang(\A_q)$ for all $ i\geq 1$, yielding infinitely many distinguishing words.
This contradicts the assumption that $p \approx q$, so the lemma follows.

Moreover, this lemma yields a brute-force algorithm for checking whether two states $p$ and $q$ are almost equivalent: enumerate all words of length up to $n^2$. 
By \Cref{lem:dfa-word-bound}, we have $p \approx q$ if and only if $p$ and $q$ differ only on words of length smaller than $n^2$.
We briefly mention this algorithm for comparison with the almost-equivalence computation for DRAs. 
In fact, more efficient algorithms exist for computing the almost-equivalence relation over all states of DFAs, running in $\bigO(n^2)$~\cite{Badr:08:CIAA} and even $\bigO(n \log n)$~\cite{Holzer:10:CIAA}.

We fix a \textbf{well-typed} \dra $\A = (Q, q_0, F, \Delta)$ with $n$ states and $k$ registers, and we want to hyper-minimize it in the remainder of the paper.

Our hyper-minimization algorithm for DRAs follows the general approach used for DFAs. However, extending hyper-minimization to DRAs is considerably more involved, as identifying suitable analogues of the key notions is non-trivial. 
We identify the two main challenges as follows, and address them in Sections~\ref{ssec:almost-eq-relation-states} and~\ref{sec:hyper-minimization_process}, respectively.
\begin{itemize}
    \item[C1] How to define and compute the almost-equivalence relation over states given that a run on a word of a \dra is not a sequence of states but configurations?

    \item[C2] How to design and apply merge operations, and ensure that the final resulting \dra remains both hyper-data-minimal and hyper-state-minimal?
\end{itemize}


\subsection{Almost-Equivalence Relation over States}
\label{ssec:almost-eq-relation-states}

In this section, we introduce the almost-equivalence relation over states and then provide the means to compute them over dense domains or non-ordered domains.
Firstly, we define the \emph{almost-equivalent} relation over configurations.

\begin{definition}\label{def:dra-almost-eq-configurations}
    Let $(p, u)$ and $(q, v)$ be two configurations of $\A$. We say they are \emph{almost-equivalent}, denoted $(p, u) \approx (q, v)$, if there exist only finitely many word types $w \in \finwords$ s.t. $(p, u) \xrightarrow{w} (p', u')$ and $(q, v) \xrightarrow{w} (q', v')$ with $p' \in F \not\Leftrightarrow q' \in F$.
\end{definition}

Now, we are ready to define the DRA analogue of the almost-equivalence relation over states.

\begin{definition}\label{def:dra-almost-eq-states}
Let $p$ and $q$ be two states in $\A$.
$p$ and $q$ are \emph{almost-equivalent}, denoted by $p \approx q$, if for each register word $u$ of $p$, there exists a register word $v$ of $q$ such that $(p,u) \approx (q,v)$, and vice versa.
\end{definition}

Two configurations $(p, u)$ and $(q, v)$ are \emph{equivalent}, denoted by $(p, u) \equiv (q, v)$, if $\lang(\autA_{(p, u)}) = \lang(\autA_{(q, v)})$. This equivalence relation $\equiv$ can be extended to states, analogously to almost-equivalence (cf. \Cref{def:dra-almost-eq-states}).
Moreover, both almost-equivalence and equivalence relations naturally extend to configurations or states of \emph{different} \dra{s} over the same data domain $(\Sigma,R)$.
For example, for configurations $(p,u)$ of $\autA$ and $(q,v)$ of $\autB$, we have $(p,u) \approx (q,v)$ if there are finitely many word types $w \in \finwords$ such that $(p,u) \xrightarrow{w} (p',u')$ and $(q,v) \xrightarrow{w} (q',v')$ with $p' \in F_\autA \not\Leftrightarrow q' \in F_\autB$, where $F_\autA$ and $F_\autB$ denote the sets of accepting states of $\autA$ and $\autB$, respectively. All properties of almost-equivalence extend to this setting.

To compute the almost-equivalence relation over states, we develop the following DRA analogue of Lemma~\ref{lem:dfa-word-bound}.
\begin{restatable}{lemma}{lemDraWordBound}\label{lem:dra-word-bound}
Let $(p,u),(q,v)$ be two configurations of a \dra $\autA$. $(p,u) \approx (q,v)$, iff there exists $\ell \in \naturals$ s.t. for all $w \in \finwords$ with $|w| > \ell$, the configurations reached from $(p,u)$ and $(q,v)$ after reading $w$ either both accepting or both rejecting. 
\end{restatable}

Using \Cref{lem:dra-word-bound}, one can prove that the successors of almost-equivalent configurations over a given letter are also almost-equivalent.
\begin{restatable}{lemma}{lemAlmostEqToAlmostEq}\label{lem:almost-eq_to_almost-eq}
    Let $\autA$ be a  \dra over $(\Sigma,R)$, and let $c$ and $d$ be two configurations with $c \approx d$. 
    For every $a \in \Sigma$, if $c \xrightarrow{a} f$ and $d \xrightarrow{a} g$, then $f \approx g$.
\end{restatable}

Further, by combining \Cref{lem:pump_two_runs} and \Cref{lem:dra-word-bound}, we immediately derive:
\begin{corollary}\label{lem:almost-eq_with_lenTA_bound}
Let $\autA$ be a  \dra and $(p,u),(q,v)$ be two configurations of $\autA$. $(p,u) \approx (q,v)$, if and only if for all $w \in \Sigma^{> \lenTA}$, the configurations reached from $(p,u)$ and $(q,v)$ after reading $w$ either both accepting or both rejecting.
\end{corollary}
Note that if $(p,u)$ and $(q,v)$ in the lemma are from different \dra{s} $\autA$ and $\autB$, where $\autA$ has $m$ states and $k$ registers and $\autB$ has $n$ states and $\ell$ registers, then the bound $\lenTA$ should be replaced with $m \cdot n \cdot (k + \ell)!$.

The equivalence problem for \dra{s} is decidable because it can be reduced to the emptiness problem using the product automaton technique; emptiness is itself decidable for \dra{s} \cite{Kaminski:94:TCS,Ley:10:TR}. Note that if an automaton recognizes the empty language, its minimal form consists of a single rejecting state. The minimization problem for \dra{s} is also decidable \cite{Ley:10:TR}.
So by \Cref{lem:almost-eq_with_lenTA_bound}, the configuration almost-equivalence problem is also decidable if we can compute \dra{s} recognizing $\lang(\autA_{(p,u)}) \cap \Sigma^{\geq \ell}$ and $\lang(\autA_{(q,v)}) \cap \Sigma^{\geq \ell}$. 
Let $\A_{\ell}$ be a DFA that recognizes the language $\Sigma^{\geq \ell}$.
Consequently, the product automaton of $\autA_{(p,u)}$ and $\autA_\ell$ is a \dra that recognizes $\lang(\autA_{(p,u)}) \cap \Sigma^{\geq \ell}$. Analogously, the product automaton of $\autA_{(q,v)}$ and $\autA_\ell$ is a \dra recognizing $\lang(\autA_{(q,v)}) \cap \Sigma^{\geq \ell}$. As a result, we have:

\begin{lemma}\label{lem:decidability_config_almost-eq}
    Let $\autA$ be a \dra, and let $c, d$ be two configurations of $\autA$. It is decidable to determine whether $c \approx d$.
\end{lemma}

By Definition~\ref{def:dra-almost-eq-states}, when checking whether two states are almost-equivalent, it seems that we need to enumerate all possible register words of the two states (which are infinite), and check the almost-equivalence relation over configurations using Lemma~\ref{lem:dra-word-bound}.
Nevertheless, with Lemma~\ref{lem:almost-eq-configuration_imp_almost-eq-states}, we only need to consider one configuration for a state.

\begin{restatable}{lemma}{lemAlmostEqConfigurationImpAlmostEqStates}\label{lem:almost-eq-configuration_imp_almost-eq-states}
    Let $\autA$ be a \dra, and let $(p,u)$ and $(q,v)$ be two configurations of $\autA$. If $(p, u) \approx (q, v)$ (resp., $(p,u) \equiv (q,v))$, then $p \approx q$ (resp., $p\equiv q$).
\end{restatable}

Let $p$ and $q$ be two states of $\autA$ with register types $\tau$ and $\eta$, respectively, and let $u$ be a word of type $\tau$. By \Cref{lem:almost-eq-configuration_imp_almost-eq-states}, we have $p \approx q$ if and only if there exists a word $v$ of type $\eta$ such that $(p,u) \approx (q,v)$. Moreover, by \Cref{lem:decidability_config_almost-eq}, almost-equivalence of configurations is decidable. Hence, to decide whether $p \approx q$, it suffices to determine whether there exists such a word $v$.

The main difficulty lies in the fact that there are infinitely many candidates for $v$. Our key observation is that it is unnecessary to consider all such words. Instead, one can effectively compute a \emph{finite} set $W_\eta \subset \alphabet$ such that there exists a word $v$ with $(p,u) \approx (q,v)$ if and only if there exists $v'$ over $ W_\eta$ satisfying $(p,u) \approx (q,v')$.
The following lemma helps define the finite set $W_\eta$.

\begin{restatable}{lemma}{lemAlmostEqSymbolsRange}\label{lem:almost-eq_symbols_range}
    Let $\autA$ be a  \dra over a dense domain and $(p,u), (q,v)$ be two configurations of $\autA$. 
        %
    Then, if $(p,u) \approx (q,v)$, then for every word $v' \in \finwords$ with $uv\sim_R uv'$, we have $(p,u) \approx (q,v')$.
\end{restatable}

To examine whether there exists a register assignment $v$ for a given configuration $(p,u)$ and a state $q$ such that $(p,u) \approx (q,v)$, \Cref{lem:almost-eq_symbols_range} suggests that we only need to consider the word type of the concatenation $u \cdot v$; the exact values in $v$ do not affect the result.
Therefore, given a word $u $ and $v$ of $\eta$-type with $|\eta| = k$, to compute the alphabet $W_{\eta}$ of $v'$ that cover all possible types of $u\cdot v$ (and thus $uv'$), i.e., $uv\sim_R uv'$, we do the following:
\begin{enumerate}
    \item Let $a_1, a_2, \cdots , a_m$ be the increasing sequence of the letters in $u$.
    Let $a_0 = -\infty$ and $a_{m+1} = +\infty$.

    \item Let $b_{i,1}, b_{i,2}, \cdots, b_{i,k}$ be $k$ fresh letters in increasing order such that $a_i < b_{i,j} < a_{i+1}$ for all $0\leq i \leq m$ and $1\leq j \leq k$. 
    \item Then, we define $W_{\eta} = u\cup \bigcup_{i=0}^m \{b_{i,1}, \cdots, b_{i,k}\}$.
\end{enumerate}
Note that if the binary relation $R$ is identity test, then we only need to add all letters in $u$ together with $k$ fresh letters into $W_{\eta}$.

Hence, according to \Cref{lem:almost-eq_symbols_range}, it is not hard to see:
\begin{corollary}\label{coro:relative-alphabet}
    Let $p$ has register type $u$ and $q$ has register type $\eta$.
    Let $W_{\eta}$ be the finite set computed above.
    Then, there exists a word $v$ such that $(p, u) \approx (q, v)$ if, and only if, there exists a word $v'$ over $W_{\eta}$ such that $(p, u) \approx (q, v')$.
\end{corollary}
    
We now give an example to explain how to compute $W_{\eta}$.
Consider a DRA $\A$ over $(\mathbb{Q}, <)$:
let $p$ and $q$ be two states where $p$ has a register type $\tau = 1 \cdot 2$ while $q$ has the register type $\eta = 2 \cdot 1$.
Recall that $\tau$ and $\eta$ represent the relative order of registers.
Let $(p, u)$ be a configuration of $p$ where $u = 3 \cdot 7$.
To cover all possible configurations $(q, v')$ associated with $q$, it suffices to consider valuations $v$ over the finite alphabet
$
W_{\eta} = \{3, 7\} \cup \{1, 2\} \cup \{4, 5\} \cup \{9, 11\},
$
constructed as follows.
Let $a_1 = 3$ and $a_2 = 7$ in $u$, and define $a_0 = -\infty$ and $a_3 = +\infty$.
\begin{itemize}
    \item Between $a_0 = -\infty$ and  $a_1 = 3$, we pick $b_{0, 1} = 1$ and $b_{0,2} = 2$.
    \item Between $a_1 = 3$ and  $a_2 = 7$, we select $b_{1, 1} = 4$ and $b_{1,2} = 5$.
    \item Between $a_2 = 7$ and $a_3 = +\infty$, we choose $b_{2,1} = 9$ and $b_{2,2} = 11$.
\end{itemize}

Let $p$ and $q$ be two states in $\A$ where $\tau$ is the register type of $p$ and $\eta$ is the register type of $q$.
To check whether $p \approx q$ holds, we perform the following steps:
\begin{enumerate}
    \item Find a word $u$ of type $\tau$.
    \item Compute the finite set $W_{\eta}$ based on $u$ and $\eta$.
    \item Enumerate all $v'$ of type $\eta$ over $W_{\eta}$ and check whether $(p, u) \approx (q, v')$.
    \item If there is a word $v'$ with $(p, u) \approx (q, v')$, then we conclude that $p \approx q$;
    otherwise, we know that $p \not\approx q$.
\end{enumerate}

 Since words of a given type over the finite set $W_{\eta}$ can be enumerated and, by \Cref{lem:decidability_config_almost-eq}, the configuration almost-equivalence problem for \dra{s} is decidable, we immediately obtain the following result:

\begin{theorem}\label{thm:almost-eq_decidability}
The state almost-equivalence problem for  \dra{s} is decidable.
\end{theorem}

\subsection{The Hyper-Minimization Process for \dra{s}}\label{sec:hyper-minimization_process}


\begin{algorithm}[t]
\caption{Hyper-minimization of a \dra $\mathcal{A}$. 
}
\label{algo:hyper-minimization}
    $\A=\textit{minimize}(\mathcal{A})$\;\label{line:minimal}

    $\A=\datamin(\mathcal{A})$\; \label{line:data_minimal}
    \While{$\exists p\neq q \text{ with the same register type}. (p\approx q \text{ and } p\text{ is a preamble} ) \vee (p\equiv q) $\label{line:merge_start}}{
        $\autA = \merge(p \rightarrow q, \autA)$;\label{line_merge}
    }\label{line:merge_end}
    return $\mathcal{A}$\;
\end{algorithm}
\begin{algorithm}[t]
\caption{The \datamin~step of a minimal \dra $\autA$.}
\label{algo:hyper-data-minimization}
    \For{\text{preamble state } $q\in Q$ \text{ and } $w$ \text{ is a shortest word leading } $\autA$ \text{ to } $q$}{
     Let $u$ be such that $(q_0, \emptyword) \xrightarrow{w} (q, u)$ in the original well-typed DRA\;
        \For{symbol $a\in u$}{
            \If{$a$ is not $\lenTA$-memorable}{
                remove $a$ from $u$\;
            }
        }
        make $q$ $u$-register-type\;
    }
\end{algorithm}

With the check for the almost-equivalence relation between two states established, we now present \Cref{algo:hyper-minimization} to hyper-minimize a given \dra $\autA$. For ease of explanation, we assume that every preamble state can reach an accepting state via infinitely many distinct paths in the underlying graph. Otherwise, such a state can simply be removed from the automaton, as the resulting language remains almost-equivalent to the original.

Similar to the process for DFAs, we first minimize the \dra. However, before entering the loop to merge almost-equivalent states, we perform a {\datamin} step to ensure that only those symbols that can affect the acceptance of infinitely many extensions are kept in the registers. The details of this {\datamin} step are provided in \Cref{algo:hyper-data-minimization}. We then check whether there exist almost-equivalent states $p$ and $q$ such that $p$ is a preamble state and $p,q$ share the \emph{same} register type; if so, they can be merged.



First, consider the process of hyper-data-minimization. Recall that a symbol in a word $x$ is memorable if its appearance in some word extension $w$ is relevant to whether $x \cdot w$ is accepted. Furthermore, in a \dra recognizing the language, every memorable symbol must be kept in the registers. However, there are no restrictions on the length of $w$. On the contrary, by \Cref{lem:almost-eq_with_lenTA_bound}, we have that two configurations $(p,u)$ and $(q,v)$ of a given \dra $\autA$ are almost-equivalent if and only if for every word $w$ of length greater than $\lenTA$, the runs from these configurations on $w$ are either both accepting or both rejecting. This suggests that if a symbol $a \in u$ is not relevant to the acceptance of any word $w$ with length greater than $\lenTA$, then a \fullhdm \dra for $\autA$ does not need to keep $a$ in its registers. To formalize this, we introduce the following definition:

\begin{definition}[$\ell$-Memorable Letters]
    Let $\autA$ be a \dra, $(p,u)$ be a configuration, and $\ell \in \mathbb{N}$.
    A symbol $a \in u$ is called $\ell$-memorable if there exist a word $w \in \Sigma^{\ge \ell}$ and a symbol $b$ such that $w \sim_R w[a/b]$, and the configurations reached from $(p,u)$ after reading $w$ and $w[a/b]$, respectively, are such that exactly one of them is accepting.
\end{definition}

The notion of $\ell$-memorability identifies whether a specific data value stored in a register is essential for the future acceptance of word extensions of length at least $\ell$. Since the hyper-data-minimization process relies on removing symbols that are not $\lenTA$-memorable, it is necessary to ensure that this property can be effectively checked.
Note the distinction from the notion of $L$-memorability (cf.~\Cref{def:memorable}): here, $l$ is a parameter.

\begin{restatable}{lemma}{lemDecidabilityOfEllMemorability}\label{lem:decidability_of_ell_memorability}
    Let $\autA$ be a \dra, $(p,u)$ a configuration of $\autA$, and $\ell \in \mathbb{N}$. It is decidable whether a symbol in $u$ is $\ell$-memorable.
\end{restatable}

By \Cref{lem:almost-eq_with_lenTA_bound}, two configurations are almost-equivalent if and only if their runs on every word extension of length greater than $\lenTA$ either both accept or both reject. One might therefore assume that the hyper-minimization process aims to produce an automaton that does not retain any non-$\lenTA$-memorable symbols. However, given a configuration $(p,u)$, if $p$ is a kernel state and a symbol $a \in u$ is memorable, then $a$ must be retained regardless of whether it is $\lenTA$-memorable. Otherwise, by the definitions of kernel states and memorability, the resulting automaton would not be almost-equivalent to the original. Accordingly, we have the following lemma:

\begin{lemma}\label{lem:hyper-data-minimal_bound}
    Let $\autA$ be a minimal \dra. Then, $\autA$ is \fullhdm if and only if for every configuration $(p,u)$ of $\autA$ where $p$ is a preamble state, every symbol $a \in u$ is $\lenTA$-memorable.
\end{lemma}

Recall our assumption that all preamble states can reach an accepting state via infinitely many distinct paths. Without this assumption, there could exist a configuration $(p,u)$ in a \fullhdm $\autA$ where $p$ is a preamble state and $a \in u$. In such a case, the symbol $a$ would not be $\lenTA$-memorable simply because the language recognized from $p$ contains only finitely many word types; consequently, none of the values in $u$ would be $\lenTA$-memorable. See \Cref{app:ex_preamble_assumption} for an example.

Intuitively, our hyper-data-minimization process ensures that the automaton does not retain any non-$\lenTA$-memorable symbols for preamble states. However, when a non-$\lenTA$-memorable symbol is removed from a state's register type, the automaton may become non-deterministic. For instance, suppose state $p$ has register type $1\cdot 3 \cdot 5$, and there are two transitions $(p, 1\cdot 3 \cdot 5\cdot 2, E, q)$ and $(p, 1\cdot 3\cdot 5\cdot 4, E', q')$. If we remove the value at index $3$ from the types $1\cdot 3 \cdot 5$, $1\cdot 3\cdot 5\cdot 2$, and $1\cdot 3\cdot 5\cdot 4$, then $p$ may have non-deterministic outgoing transitions. To resolve this problem, we require the following lemma:

\begin{restatable}{lemma}{lemNondeterministicRemovable}\label{lem:nondeterminitic_removeble}
    Let $p$ be a preamble state, $(p,u)$ a configuration of a \dra $\autA$ over $(\Sigma,R)$ with $a \in u$, and let $v$ be the register assignment obtained by removing $a$ from $u$. Let $\delta = (p, \tau, E, q)$ and $\delta' = (p, \tau', E', q')$ be two outgoing transitions from $p$, where $\tau = u \cdot b$ and $\tau' = u \cdot b'$ such that $v \cdot b \sim_R v \cdot b'$. If $a$ is not $\lenTA$-memorable in $u$, then the automaton obtained by removing $\delta$ or $\delta'$ and setting $p$ to be $v$-typed is almost-equivalent to $\autA$.
\end{restatable}

Using \Cref{lem:hyper-data-minimal_bound}, we can derive the hyper-data-minimization process for \dra{s}, as shown in \Cref{algo:hyper-data-minimization}. For each preamble state $q$:
\begin{itemize}
    \item Obtain a representative word $w$ such that $(q_0,\epsilon) \xrightarrow{w} (q,u)$ for some configuration $u$ in the original input DRA. 
    \item Check every symbol $a \in u$ to determine whether $a$ is $\lenTA$-memorable, which is decidable by \Cref{lem:decidability_of_ell_memorability}.
    \item For every incoming transition $(p,\tau,E,q)$ to $q$, modify the set $E$ s.t. any symbol that is not $\lenTA$-memorable is no longer kept in the registers upon entering $q$.
    \item Update the word type $\eta$ of every outgoing transition $(q,\eta,G,r)$ from $q$. By \Cref{lem:nondeterminitic_removeble}, if there is more than one such transition, we can arbitrarily retain one of them.
\end{itemize}

Next, we consider the merging of almost-equivalent states $p$ and $q$ in \Cref{algo:hyper-minimization}, starting from 
\Cref{line:merge_start}. 
There are two cases in which we cannot merge $p$ and $q$: (i) Both $p$ and $q$ are kernel. (ii) $p$ and $q$ have different register types.

Case (i) is also considered in DFA hyper-minimization. Let $\lang_{q_0\rightarrow p}$ be the set of words for which the run from the initial configuration ends in $p$. The set $\lang_{q_0\rightarrow p}$ is finite if and only if $p$ is a preamble state. Since $p \approx q$, the set $\lang_{q_0\rightarrow p} \cdot (\lang(\autA_p) \symdiff \lang(\autA_q))$ contains only finitely many word types if and only if $\lang_{q_0\rightarrow p}$ contains only finitely many word types. Accordingly, we can merge $p$ into $q$ only if $p$ is a preamble state; otherwise, the resulting automaton would not be almost-equivalent to the original.

Consider Case (ii). If $(p,u) \approx (q,v)$, \Cref{lem:hyper-data-minimal_bound} implies that for every symbol $a$, $a \in u \Leftrightarrow a \in v$. However, if $p$ and $q$ have different register types, then $u \not\sim_R v$. Suppose $w, x, y$ are three words such that $|w| > \lenTA$, $(q_0,\epsilon) \xrightarrow{x} (p,u) \xrightarrow{w} c$, and $(q_0,\epsilon) \xrightarrow{y} (q,v) \xrightarrow{w} d$ for some configurations $c$ and $d$. Since $(p,u) \approx (q,v)$, $c$ and $d$ are either both accepting or both rejecting. Because all values in $u$ and $v$ are $\lenTA$-memorable, we can assume w.l.o.g. that modifying these values would cause $c$ and $d$ to switch their acceptance status (from accepting to rejecting, or vice versa). Furthermore, since $u \not\sim_R v$, we cannot redirect transitions ending in $p$ to $q$ without eliminating values in $u$ and $v$. Accordingly, any such elimination would result in an automaton that is no longer almost-equivalent to the original.

\begin{restatable}{lemma}{lemDraToAlmostEqDra}\label{lem:dra_to_almost-eq_dra}
    Let $\autA$ be a \dra, and $p, q$ be two almost-equivalent states of $\autA$ with the same register type. If $p$ is preamble, then the \dra $\autB$ obtained by performing the merging operation $\merge(p \rightarrow q, \autA)$ is almost-equivalent to $\autA$.
\end{restatable}

So, according to \Cref{lem:dra_to_almost-eq_dra}, at the end of \Cref{line:merge_end} of \Cref{algo:hyper-minimization}, the resulting \dra is almost-equivalent to the input \dra.

\begin{restatable}{lemma}{lemMappingBwAlmostEqDras}\label{lem:mapping_bw_almost-eq_dras}
    For every two almost-equivalent \fullhdm \dra{s} $\mathcal{A}$ and $\mathcal{B}$ with sets $Q_\mathcal{A}$ and $Q_\mathcal{B}$ of states, respectively, 
    there exists a function $h: Q_\mathcal{A}\rightarrow Q_\mathcal{B}$ such that for all $q\in Q_\mathcal{A}$:
    \begin{enumerate}[label=(a)]
        \item\label{lem:mapping_bw_almost-eq_dras:prop:a} $q\approx h(q)$ and $q, h(q)$ have the same register type.
        \item\label{lem:mapping_bw_almost-eq_dras:prop:b} $h(q)$ is a kernel state if $q$ is a kernel state.
    \end{enumerate}
\end{restatable}

Suppose that $\autA$ is the  \dra obtained at the end of \Cref{line:merge_end} of \Cref{algo:hyper-minimization}, $\autB$ is a \fullhsm \dra, and $h$ is the mapping satisfying \Cref{lem:mapping_bw_almost-eq_dras}. For every $p, q \in Q_\autA$, if $h(p)=h(q)$, then by \Cref{lem:mapping_bw_almost-eq_dras:prop:a}, we have $p \approx q$, and they have the same register type. If at least one of $p$ or $q$ is a preamble state, then they must be the same state (i.e., $p=q$); otherwise, they would have been merged during the execution of \Cref{algo:hyper-minimization}. If both are kernel states, then by \Cref{lem:almost-eq-configuration_imp_almost-eq-states} and \Cref{lem:mapping_bw_almost-eq_dras:prop:b}, it follows that $p \equiv q$. In this case, we must also have $p=q$, as they otherwise would have been merged by the algorithm. Accordingly, $h$ is an injective mapping, and the \dra resulting from \Cref{line:merge_end} of \Cref{algo:hyper-minimization} is a \fullhsm  \dra for the given input.
Consequently, we obtain the following result:

\begin{theorem}
    For every input  \dra $\autA$, the output \dra of \Cref{algo:hyper-minimization} is a \fullhm  \dra for the language $\lang(\autA)$.
\end{theorem}

\section{Conclusion}



We introduced a notion of hyper-minimization for DRAs and presented an algorithm to compute hyper-minimal DRAs. To the best of our knowledge, this is the first work addressing hyper-minimization in this setting.

Our algorithm follows the general structure of classical hyper-minimization for DFAs: it checks almost-equivalence between states and performs merge operations. Extending this approach to DRAs, however, is technically challenging. In particular, it requires an additional hyper-data-minimization step, and all components—hyper-data-minimization, almost-equivalence checking, and state merging—must be carefully designed to ensure that the resulting automaton is both hyper-data-minimal and hyper-state-minimal.
Finally, we note that, as in the DFA setting, hyper-minimal DRAs are \emph{not} unique.


\bibliographystyle{abbrv}
\bibliography{refs}

\newpage
\appendix
\section{Missing Proofs of \Cref{sec:dra_properties}}

\subsection{Proof of \Cref{lem:pump_two_runs}}

\LemPumpTwoRuns*
\begin{proof}
    Let $w$ be a word of length greater than $\lenTA$, and let $\pi: (p,u) \xrightarrow{w} (r,x)$ and $\rho: (q,v) \xrightarrow{w} (s,y)$. By the pigeonhole principle, there exists a decomposition $w = \nu \omega_0 \mu$ with $\omega_0$ non-empty such that:
    \[ \pi = (p,u) \xrightarrow{\nu} (c, \alpha_0) \xrightarrow{\omega_0} (c, \alpha_1) \xrightarrow{\mu} (r,x) \]
    and
    \[ \rho = (q,v) \xrightarrow{\nu} (d, \beta_0) \xrightarrow{\omega_0} (d, \beta_1) \xrightarrow{\mu} (s,y) \]
    for some configurations $(c, \alpha_0), (c, \alpha_1), (d, \beta_0)$, and $(d, \beta_1)$ where $\alpha_0 \cdot \beta_0 \sim_R \alpha_1 \cdot \beta_1$. 

    Therefore, there exists an order-preserving mapping $\sigma$ such that $\sigma(\alpha_0) = \alpha_1$ and $\sigma(\beta_0) = \beta_1$. Let $\alpha_{i+1} = \sigma(\alpha_i)$, $\beta_{i+1} = \sigma(\beta_i)$, and $\omega_{i+1} = \sigma(\omega_i)$ for all $i \in \mathbb{N}$. Since $\sigma$ is order-preserving, we have $\sigma(\alpha_i \cdot \beta_i) \sim_R \alpha_i \cdot \beta_i$ for all $i$, and thus $\alpha_i \cdot \beta_i \sim_R \alpha_{i+1} \cdot \beta_{i+1}$ for all $i$. 

    Accordingly, 
    \[ (p,u) \xrightarrow{\nu} (c, \alpha_0) \xrightarrow{\omega_0} (c, \alpha_1) \dots \xrightarrow{\omega_{i-1}} (c, \alpha_i) \xrightarrow{\sigma^i(\mu)} (r, \sigma^i(x)) \]
    and
    \[ (q,v) \xrightarrow{\nu} (d, \beta_0) \xrightarrow{\omega_0} (d, \beta_1) \dots \xrightarrow{\omega_{i-1}} (d, \beta_i) \xrightarrow{\sigma^i(\mu)} (s, \sigma^i(y)) \]
    are valid runs on the word $\nu \cdot \omega_0 \dots \omega_{i-1} \cdot \sigma^i(\mu)$ for all $i \in \mathbb{N}$, where $\omega_0 \dots \omega_{i-1} = \epsilon$ if $i=0$. Consequently, we derive the statement.
    \qed
\end{proof}

\section{Hyper-Minimization for DFAs}
\label{app:dfa-hypermin}
\begin{algorithm}[htb]
\caption{Hyper-minimization of a DFA $\autA$.}
\label{algo:dfa-hyper-minimization}
    $\autA = \textit{minimize}(\mathcal{A})$\;
    \While{$\exists p\neq q. p\approx q \text{ and }p\text{ is a preamble state}$}{
        $\autA = \merge(p \rightarrow q, \autA)$\;
    }
    return $\autA$\;
\end{algorithm}

\section{Missing Proofs of \Cref{sec:hyper-minimization}}

\subsection{Proof of \Cref{lem:dra-word-bound}}

\lemDraWordBound*
\begin{proof}
    First, consider the ``if'' direction. Suppose that for all $w \in \finwords$ with $|w| \ge \ell$, the configurations reached from $(p,u)$ and $(q,v)$ after reading $w$ are either both accepting or both rejecting. Since the number of word types of lengths up to $\ell$ is finite, there are only finitely many word types $w$ such that $(p, u) \xrightarrow{w} (p', u')$ and $(q, v) \xrightarrow{w} (q', v')$ with $p' \in F \not\Leftrightarrow q' \in F$. Thus, $(p,u) \approx (q,v)$.

    Now, consider the ``only-if'' direction. Assume that $(p, u) \approx (q, v)$. By \Cref{def:dra-almost-eq-configurations}, there are finitely many word types $w$ such that $(p, u) \xrightarrow{w} (p', u')$ and $(q, v) \xrightarrow{w} (q', v')$ with $p' \in F \not\Leftrightarrow q' \in F$. Let $\ell$ be the maximum length among these word types. Then, if $(p, u) \xrightarrow{w} (r, x)$ and $(q, v) \xrightarrow{w} (s, y)$ for any $w$ with $|w| > \ell$, it must be that $r$ and $s$ are either both accepting or both rejecting.
    \qed
\end{proof}

\subsection{Proof of \Cref{lem:almost-eq_to_almost-eq}}

\lemAlmostEqToAlmostEq*
\begin{proof}
    Assume by contraposition that $f \not\approx g$. Let $\ell$ be the bound satisfying \Cref{lem:dra-word-bound}. Since $f \not\approx g$, there exists a word $w$ with $|w| > \ell$ such that exactly one of the runs from $f$ and $g$ on $w$ is accepting. This implies that $aw$ is accepted starting from exactly one of $c$ and $d$. Since $|aw| > \ell$, this contradicts the fact that $c \approx d$.
    \qed
\end{proof}

\subsection{Proof of \Cref{lem:almost-eq-configuration_imp_almost-eq-states}}

\lemAlmostEqConfigurationImpAlmostEqStates*
\begin{proof}
    Let $u'$ have the same type as $u$. Then there exists an order-preserving mapping $\sigma$ such that $\sigma(u)=u'$. Let $v'=\sigma(v)$. Given $w \in \Sigma^{> \lenTA}$, by \Cref{lem:pump_two_runs} and \Cref{lem:dra-word-bound}, the configuration reached from $(p,u')$ after reading $w$ is accepting iff the configuration reached from $(p,u)$ after reading $\sigma^{-1}(w)$ is accepting. Since $(p, u) \approx (q, v)$, the configurations reached from $(p,u)$ and $(q,v)$ after reading $\sigma^{-1}(w)$ are either both accepting or both rejecting.
    
    Similarly, the configuration reached from $(q,v')$ after reading $w$ is accepting iff the configuration reached from $(q,v)$ after reading $\sigma^{-1}(w)$ is accepting. As a result, $(p,u') \approx (q,v')$. By the same argument, for every word $v'$ of the same type as $v$, there exists a word $u'$ of the same type as $u$ such that $(q,v') \approx (p,u')$, and vice versa, which then entails that $p \approx q$.
    \qed
\end{proof}

\subsection{Proof of \Cref{lem:almost-eq_symbols_range}}

\lemAlmostEqSymbolsRange*
\begin{proof}
    Let $v' $ be a word such that $u\cdot v \sim_R u\cdot v'$.
    Let $\sigma$ be an order-preserving mapping such that $\sigma(u\cdot v')=u\cdot v$.
    Since $(p,u) \approx (q,v)$, by \Cref{lem:dra-word-bound}, there exists $\ell$ such that for every $w \in \Sigma^{\geq \ell}$, the configurations reached by $\autA$ on $w$ from $(p,u)$ and $(q,v)$ are either both accepting or rejecting.
    Let $w' \in \alphabet^{\geq \ell}$.
    We then have $\sigma(w') \in \alphabet^{\geq \ell}$ as well.
    It then follows that the configurations reached by $\autA$ on $\sigma(w')$ from $(p, u)$ and $(q, v)$ are either both accepting or rejecting. 
    Further, by applying the mapping $\sigma^{-1}$, we have that the configurations reached by $\autA$ on $\sigma^{-1}(\sigma(w')) = w'$ from $(p, \sigma^{-1}(u)) = (p, u)$ and $(q, \sigma^{-1}(v)) = (q, v')$ are either both accepting or rejecting. 
    Therefore, according to \Cref{lem:dra-word-bound}, $(p, u) \approx (q, v')$.
    \qed
\end{proof}

\subsection{Proof of \Cref{lem:decidability_of_ell_memorability}}

\lemDecidabilityOfEllMemorability*
\begin{proof}
    Suppose that $u=a_1\dots a_k$. If $a_i=a$ is $\ell$-memorable, there exist a word $w \in \Sigma^{\geq \ell}$ and a symbol $b$ such that $w\sim_R w'=w[a/b]$, and starting from $(p,u)$, exactly one of $w$ and $w'$ is accepted. 

    If $|w| > \lenTA$, there must be a repetition of the pair of states reached in the runs starting from $(p,u)$ on $w$ and $w'$. Specifically, let the runs be $\pi: (p,u) \xrightarrow{w} (q,v)$ and $\pi': (p,u) \xrightarrow{w'} (q',v')$. By the pigeonhole principle, there exist decompositions $w=w_1 w_2 w_3$ and $w'=w'_1 w'_2 w'_3$ with $|w_j|=|w'_j|$ for $j\in\{1,2,3\}$ and $w_2, w'_2$ non-empty, such that:
    \[ \pi = (p,u) \xrightarrow{w_1} (r,x) \xrightarrow{w_2} (r,y) \xrightarrow{w_3} (q,v) \]
    \[ \pi' = (p,u) \xrightarrow{w'_1} (s,x') \xrightarrow{w'_2} (s,y') \xrightarrow{w'_3} (q',v') \]
    where the state pairs $(r,s)$ at the start and end of $w_2$ and $w'_2$ are identical and $xx' \sim_R yy'$. Since $\autA$ is over a dense ordered domain or a non-ordered one, there exists an order-preserving mapping $\sigma$ such that $\sigma(y)=x$ and $\sigma(y')=x'$. By applying this mapping to the suffix, we can obtain shorter runs ending in $(q, \sigma(v))$ and $(q', \sigma(v'))$, respectively.

    As a result, we can assume $|w| \leq \max\{\ell, \lenTA\}$. Analogously to the construction of the alphabet for the set of words $W_\eta$ in \Cref{coro:relative-alphabet}, we can construct a finite alphabet $\Sigma'$ such that $a$ is $\ell$-memorable over $\Sigma$ if and only if it is $\ell$-memorable over $\Sigma'$. Since both the alphabet and the word length are finite, this property is decidable.
    As a result, we only need to consider words $w$ of length smaller than $\lenTA$ to determine whether a letter $a$ is $\ell$-memorable. As for the possible $w$, when we fix a configuration $(p, u)$ where $u$ is the register type of $p$, the possible extensions for $w$ and $b$ will be finite since we only need to cover all possible word types $uwb$, which is finite similar to the reason when we compute $W_{\eta}$ for checking almost-equivalence relation between two states. Therefore, the lemma follows.
    \qed
\end{proof}

\subsection{Proof of \Cref{lem:nondeterminitic_removeble}}

\lemNondeterministicRemovable*
\begin{proof}
    Let $w$ and $w'$ be two words of the same type starting with symbols $b$ and $b'$, respectively, both with length at least $\lenTA$. Let $\pi$ and $\pi'$ be the runs from $(p,u)$ on $w$ and $w'$, respectively. Either both runs result in acceptance or both result in rejection; otherwise, $a$ would be $\lenTA$-memorable, which is a contradiction. Accordingly, after removing $\delta$ or $\delta'$ and setting $p$ to be $v$-typed, the symmetric difference of the recognized languages consists only of words of length at most $\lenTA$.
    \qed
\end{proof}

\subsection{Proof of \Cref{lem:dra_to_almost-eq_dra}}

\lemDraToAlmostEqDra*
\begin{proof}
    For ease of explanation, we assume that $p$ is not an accepting state. Let $w \in \lang(\autA) \ominus \lang(\autB)$. Then either (1) $w \in \lang(\autA) \setminus \lang(\autB)$ or (2) $w \in \lang(\autB) \setminus \lang(\autA)$.

    First, consider Case (1). Let $\pi$ be the accepting run of $\autA$ on $w$. The only reason for $\pi$ not being a run of $\autB$ is that state $p$ is visited along $\pi$. Specifically,
    \[\pi = (q_0,\epsilon) \xrightarrow{w_1} (p,u) \xrightarrow{w_2} (f,r),\]
    where $w_1 w_2 = w$ for some word $u$ and accepting configuration $(f,r)$ in $\A$. Since $p$ is a preamble state in $\A$, there are only finitely many such word types $w_1$ leading $\A$ to $p$. Furthermore, because $p \approx q$, there are only finitely many word types $w_2$ such that $(p,u) \xrightarrow{w_2} (f,r)$ while $w_2$ is not accepted from $q$ in $\A$. Accordingly, there are only finitely many word types $w = w_1 \cdot w_2 \in \lang(\autA) \setminus \lang(\autB)$.

    Next, consider Case (2). Let $\pi$ be the accepting run of $\autB$ on $w$. The only reason for $\pi$ not being a run of $\autA$ is that a transition $(t,\tau,E,q)$ to $q$ is used, which is not a transition of $\autA$; instead, $(t,\tau,E,p)$ is a transition of $\autA$. Similarly, $\pi$ can be decomposed into two parts:
    \[\pi = (q_0,\epsilon) \xrightarrow{w_1} (q,v) \xrightarrow{w_2} (f,r),\]
    where $w_1 w_2 = w$ for some word $v$ and accepting configuration $(f,r)$, and the last transition applied along the run $(q_0,\epsilon) \xrightarrow{w_1} (q,v)$ is $(t,\tau,E,q)$. Therefore, 
    \[(q_0,\epsilon) \xrightarrow{w_1} (p,u)\]
    is a valid run of $\autA$ on $w_1$ for some $u$. Again, since $p$ is a preamble state, there are finite many word types $w_1$.
    Moreover, since $p \approx q$, there are only finitely many such word types $w_2 \in \lang(\A_q)\setminus \lang(\A_p)$.

    We thus conclude that there are only finitely many word types $w \in \lang(\A) \symdiff \lang(\autB)$.
    So, $\A$ is almost-equivalent to $\autB$.
    \qed
\end{proof}

\subsection{Proof of \Cref{lem:mapping_bw_almost-eq_dras}}

\lemMappingBwAlmostEqDras*
\begin{proof}
    For each $q \in Q_\autA$, let $[q]_\autA$ denote the class of words $w$ such that $(q^\autA_0, \epsilon) \xrightarrow{w}_\autA (q, v)$ for some $v$, where $q^\autA_0$ is the initial state of $\autA$. We define $[q]_\autB$ similarly for each $q \in Q_\autB$. Let $w_q$ be a representative word in $[q]_\autA$; we then define the mapping $h$ such that $h(q) = p$ if $w_q \in [p]_\autB$ for each $q \in Q_\autA$.

     For $\mathcal{C} \in \{\autA, \autB\}$, any state $q$ of $\mathcal{C}$, and any two words $u, v \in [q]_\mathcal{C}$, it follows from the definition of equivalence over data languages that $u \cong_{\lang(\mathcal{C})} v$.  The construction of $h$ yields $h(q^\autA_0) = q^\autB_0$. Furthermore, because $\autA \approx \autB$, we have $q^\autA_0 \approx q^\autB_0$ and, consequently, $q^\autA_0 \approx h(q^\autA_0)$. Given that the register type of the initial state is $\epsilon$, the definition of almost-equivalence over states implies that $(q^\autA_0, \epsilon) \approx (q^\autB_0, \epsilon)$.

    Now, let $p \in Q_\autA$ and $q \in Q_\autB$ such that $h(p) = q$. By definition, $(q^\autA_0, \epsilon) \xrightarrow{w_p}_\autA (p, u)$ and $(q^\autB_0, \epsilon) \xrightarrow{w_p}_\autB (q, v)$ for some $u, v$. By \Cref{lem:almost-eq_to_almost-eq}, we have $(p, u) \approx (q, v)$, and by \Cref{lem:almost-eq-configuration_imp_almost-eq-states}, it follows that $p \approx q$ and $p \approx h(p)$. By \Cref{lem:hyper-data-minimal_bound}, we have for each symbol $a$, $a \in u \Leftrightarrow a \in v$. Since both of the runs are on $w_p$, by the register updating rules, we also have $u=v$. Accordingly $p,q$ have the same register type. For any symbol $a \in \Sigma$, if $(p, u) \xrightarrow{a}_\autA (r, x)$ and $(q, v) \xrightarrow{a}_\autB (s, y)$, then $r \approx s$ must also hold. By the construction of $h$, we have $h(r) = s$. Accordingly, the statement holds by induction on the length of the words.
    \qed
\end{proof}

\section{Illustrative Examples}

\subsection{An Example of the Trade-off between the Number of Registers and the Number of States in Non-Well-Typed \dra{s}}\label{app:ex_trade-off}

\begin{example}\label{ex:inc_dec_sequences}
For each $n \in \mathbb{N}$, let $L_n$ be the language of words $w$ over $\alphabetQ$ that form a strictly increasing or decreasing sequence of length $n$.
For example, when $n=3$, the word $w = 2\cdot \text{-}1\cdot \text{-}3.5$ belongs to $L_n$, whereas $u = 1\cdot 2$ and $v = 2\cdot \text{-} 3\cdot 0$ do not.
Consider a $1$-\dra $\A_n$ that recognizes $L_n$.
It has two fashions---increasing and decreasing---and always stores the most recent input in its single register.
Given $w = a_1\dots a_m$, $\A_n$ enters the increasing (resp. decreasing) fashion if $a_2 > a_1$ (resp. $a_2 < a_1$).
For $2 < i < n$, it remains in the current fashion provided $a_i$ is larger (resp. smaller) than the value stored in the register; otherwise it rejects.
Length $m = n$ is verified using control states.
Since both fashions require $n-1$ intermediate states (plus initial and final states), $\A_n$ has $2n$ states.

While $\A_n$ is data-minimal for $L_n$, it is not state-minimal since an equivalent $(n-1)$-\dra $\A'_n$ with less states exists.
This automaton stores the first $n-1$ input symbols in its registers and checks on the fly whether these values form a strictly increasing or strictly decreasing sequence.
It needs $n+2$ states: $n+1$ states to verify the word length and one rejecting state to make $\A'_n$ complete.
\qed
\end{example}

\subsection{An Example of the Necessity of the Assumption on Preamble States}\label{app:ex_preamble_assumption}

\begin{example}
    Consider the \dra $\autA$ in \Cref{fig:DRA_mid_plus}. This \dra is a \fullhdm over $(\mathbb{Q},<)$ recognizing the language $L_{\textit{mid\_plus}} = \{a_1 a_2 \dots a_n \in \mathbb{Q}^* \mid \forall i > 2,\, a_1 < a_i < a_2\} \cup \{a_1 a_2 a_3 \in \mathbb{Q}^* \mid a_2 < a_3 < a_1\}$. Here, state $p$ is not $\lenTA$-memorable. Indeed, the automaton obtained by removing $p$ is almost-equivalent to the original, as only finitely many word types are accepted starting from $p$.
    \qed
\end{example}

\begin{figure}[]
\centering

 \begin{tikzpicture}[node distance=120pt]
 \tikzstyle{state}=[draw,shape=circle,minimum size=20pt]

 \node[state, initial] (p) at (0,0) {};
 \node[state] (r) at (2,0) {};
 \node[state,accepting] (s) at (4,1) {};
 \node[state] (sp) at (4,-1) {$p$};
 \node[state,accepting] (t) at (7,-1) {};

\path[->]  (p) edge  node[above,sloped] {$0: \emptyset$}
(r);
\path[->]  (r) edge  node[above,sloped] {$0\cdot 1: \emptyset$}
(s);
\path[->]  (r) edge  node[above,sloped] {$1\cdot 0: \emptyset$}
(sp);
\path[->]  (sp) edge  node[above,sloped] {$1\cdot 2\cdot 0: \emptyset$}
(t);
 \path[->]  (s) edge [loop above] node[above,sloped] {$0\cdot 2\cdot 1: \emptyset$}
(s);

 \end{tikzpicture}
\caption{A \fullhdm \dra $\mathcal{A}$ recognizing $L_{\textit{mid\_plus}}$.}
\label{fig:DRA_mid_plus}
\end{figure}

\section{Myhill-Nerode Theorem for \dra{s}}\label{app:myhill_nerode}

\begin{definition}[\cite{Ley:10:TR}]\label{def:eq-relation-data-languages}
Given a data language $L$ over $(\Sigma, R)$, we define an equivalence relation $\cong_L$ on $\Sigma^*$, where for every $u,v\in\Sigma^*$, we write $u \cong_L v$ if:
(1) $mem_L(u)
\sim_R mem_L(v)$, and (2) for all $x,y \in \Sigma^*$, if $mem_L(u)\cdot x \sim_R mem_L(v) \cdot y$, then $u\cdot x \in L \Leftrightarrow v\cdot y \in L$.
\end{definition}

Intuitively, equivalent words 
$u$ and $v$ must have memorable words of the same type. Furthermore, when their memorable-word extensions share the same type, extending $u$ and $v$ by the corresponding suffixes yields the same membership in $L$, since only memorable words affect future membership.

Let $L$ be a data language over $(\Sigma, R)$ and $k = \max \left\{|mem_L(u)| \mid u \in \finwords \right\}$. 
Let $[u]_{\cong_L}$ denote the equivalence class defined by $\cong_L$ where $u\in\finwords$ belongs to. 
We define the \emph{canonical $k$-\dra} 
$\mathcal{C}_L = (Q, q_0, F, \Delta)$ for $L$ as follows:
\begin{itemize}
    \item $Q = \bigcup_{i=0}^k Q_i$ where $Q_i = \left\{[u]_{\cong_L} \subseteq \finwords \mid u \in \finwords, |mem_L(u)| = i\right\}$, 
    \item $q_0 = [\emptyword]_{\cong_L}$, $F = \left\{[u]_{\cong_L} \in Q \mid u \in L\right\}$, and
    \item $([u]_{\cong_L},\tau,E, [u']_{\cong_L}) \in \Delta$ if there exists a symbol $a \in \Sigma$ such that:
    \begin{itemize}
        \item $u \cdot a \cong_L u'$, $mem_L(u\cdot a)$ has type $\tau$, and
        \item $E$ is the set of indices $i$ in the sequence
        $a_1 \dots a_d = mem_L(u) \cdot a$ such that either 
        $a_i \notin mem_L(u\cdot a)$ or $a_i = a$ for $i < d$.
    \end{itemize}
\end{itemize}
In general, $\mathcal{C}_L$ may have infinitely many states or registers and may be nondeterministic. 
When $L$ is \dra-recognizable, $\mathcal{C}_L$ is exactly the minimal \wdra for $L$.  
This Myhill–Nerode theorem for \dra{s} was established in~\cite{Ley:10:TR}:

\thmMyhill*

\end{document}